\pdfoutput=1

\documentclass[acmlarge,screen,nonacm]{acmart}
\makeatletter                   
\def\mdseries@tt{m}             
\makeatother                    
\usepackage[draft=true]{minted} 
\usepackage{color}
\usepackage{hyperref}           
\hypersetup{
    colorlinks=true,
    linkcolor=blue,
    filecolor=red,      
    urlcolor=magenta,
    breaklinks=true,            
}
\usepackage{breakurl}           

\usepackage{booktabs}   
\usepackage{subcaption} 
\usepackage[ruled,vlined,commentsnumbered]{algorithm2e}
\usepackage{minted}
\usepackage{semantic}
\usepackage{amsmath}
\usepackage{mathtools}
\usepackage[nounderscore]{syntax}
\usepackage{wrapfig}
\usepackage{makecell}
\usepackage[arrow, matrix, curve]{xy}
\usepackage{multirow}

\usepackage{mathrsfs}
\usepackage{bigdelim}

\usepackage{tikz}
\usetikzlibrary{matrix}

\usepackage{csquotes}

\makeatletter
\newcommand{\removelatexerror}{\let\@latex@error\@gobble}
\makeatother

{\tiny }

\makeatletter\if@ACM@journal\makeatother
\acmYear{2019}
\acmMonth{8}
\acmYear{2019}
\startPage{1}
\fi

\setcopyright{none}             

\makeatletter
\let\@authorsaddresses\@empty
\makeatother

\usepackage[]{todonotes}
\newboolean{showcomments}
\setboolean{showcomments}{false}
\ifthenelse{\boolean{showcomments}}
{ \newcommand{\mynote}[3]{
		\fbox{\bfseries\sffamily\scriptsize#1}
		{\small$\blacktriangleright$\textsf{\emph{\color{#3}{#2}}}$\blacktriangleleft$}}}
{ \newcommand{\mynote}[3]{}}
\newcommand{\nr}[1]{\mynote{Norman}{#1}{red}}

\newcommand{\out}[1]{}

\def\id{\text{id}}

\newcommand{\emphdef}[1]{{#1}\xspace}

\newcommand{\obj}[1]{\ensuremath{\textit{obj}({#1})}\xspace}
\newcommand{\morph}[1]{\ensuremath{\textit{morph}({#1})}\xspace}
\newcommand{\Src}{\ensuremath{\textit{src}}\xspace}
\newcommand{\src}[1]{\ensuremath{\textit{src}({#1})}\xspace}
\newcommand{\Tgt}{\ensuremath{\textit{tgt}}\xspace}
\newcommand{\tgt}[1]{\ensuremath{\textit{tgt}({#1})}\xspace}
\newcommand{\catH}{\ensuremath{\mathscr{H}}\xspace}
\newcommand{\catS}{\ensuremath{\mathscr{S}}\xspace}
\newcommand{\catC}{\ensuremath{\mathscr{C}}\xspace}
\newcommand{\freeCat}[1]{\ensuremath{\mathscr{F}({#1})}\xspace}
\newcommand{\CCat}[2]{\ensuremath{\mathscr{C}_{#1}({#2})}\xspace}
\newcommand{\calO}{\ensuremath{\mathcal{O}}\xspace}
\newcommand{\calM}{\ensuremath{\mathscr{M}}\xspace}

\newcommand{\xs}{\ensuremath{\textit{xs}}\xspace}
\newcommand{\ys}{\ensuremath{\textit{ys}}\xspace}
\newcommand{\zs}{\ensuremath{\textit{zs}}\xspace}
\newcommand{\as}{\ensuremath{\textit{as}}\xspace}
\newcommand{\bs}{\ensuremath{\textit{bs}}\xspace}
\newcommand{\cs}{\ensuremath{\textit{cs}}\xspace}
\newcommand{\ds}{\ensuremath{\textit{ds}}\xspace}
\newcommand{\flags}{\ensuremath{\textit{flags}}\xspace}

\newcommand{\smap}{\ensuremath{\textit{smap}}\xspace}
\newcommand{\map}{\ensuremath{\textit{map}}\xspace}

\newcommand{\plet}{\ensuremath{\text{let }}\xspace}
\newcommand{\pin}{\ensuremath{\text{in }}\xspace}

\newcommand{\splt}{\ensuremath{\textit{split}}\xspace}
\newcommand{\join}{\ensuremath{\textit{join}}\xspace}
\newcommand{\inl}{\ensuremath{\textit{inl}}\xspace}
\newcommand{\inr}{\ensuremath{\textit{inr}}\xspace}

\newcommand{\ourLang}{\ensuremath{\texttt{STCLang}}\xspace}

\begin{document}
\sloppy

\title
	[Category-Theoretic Foundations of STCLang]
	{Category-Theoretic Foundations of ``STCLang: State Thread Composition as a Foundation for Monadic Dataflow Parallelism''}

\author{Sebastian Ertel}
\authornote{Work done while at TU Dresden.}          
\affiliation{
  \department{Dresden Research Lab}              
  \institution{Huawei Technologies}            
  \city{Dresden}
  \country{Germany}
}
\email{sebastian.ertel@huawei.com}          

\author{Justus Adam}
\affiliation{
	\department{Chair for Compiler Construction}              
	\institution{Technische Universit\"at Dresden}            
	\city{Dresden}
	\country{Germany}
}
\email{justus.adam@tu-dresden.de}          

\author{Norman A.~Rink}
\affiliation{
	\department{Chair for Compiler Construction}              
	\institution{Technische Universit\"at Dresden}            
	\city{Dresden}
	\country{Germany}
}
\email{norman.rink@tu-dresden.de}          

\author{Andr\'es Goens}
\affiliation{
	\department{Chair for Compiler Construction}              
	\institution{Technische Universit\"at Dresden}            
	\city{Dresden}
	\country{Germany}
}
\email{andres.goens@tu-dresden.de}          

\author{Jeronimo Castrillon}
\affiliation{
	\department{Chair for Compiler Construction}              
	\institution{Technische Universit\"at Dresden}            
	\city{Dresden}
	\country{Germany}
}
\email{jeronimo.castrillon@tu-dresden.de}          

\begin{abstract}
This manuscript gives a category-theoretic foundation to the composition of State Threads as a Foundation for Monadic Dataflow Parallelism.
It serves as a supplementary formalization of the concepts introduced in the Article ``STCLang: State Thread Composition as a Foundation for Monadic Dataflow Parallelism", as published in the Proceedings of the 12th ACM SIGPLAN International Symposium on Haskell (Haskell'19)~\cite{ertel_haskell19}.
\end{abstract}

\renewcommand{\syntleft}{$}%
\renewcommand{\syntright}{$}%
\setlength{\grammarparsep}{2pt} %
\setlength{\grammarindent}{3em}%

\maketitle

\section{Category-Theoretical Foundation for State Threads}
\label{sec:cat}


This manuscripts develops in some detail a formalization of state threads in \ourLang~\cite{ertel_haskell19}.
Our development relies on the formalism of category theory.
This manuscript aims to serve as supplementary material for~\cite{ertel_haskell19}, and presumes familarity with concepts presented therein.

The two key ideas underlying \ourLang are that
(1) each state thread operates on its own private state, and %
(2) the composition of state threads retains enough information to extract parallelism from composed state threads.
Once these ideas have been made precise,\out{ in category-theoretic terms,} they naturally lead to the introduction of the \texttt{smap} functor, which generalizes \texttt{map} to situations where state must be kept track of.
The \texttt{smap} functor introduces enough structure into our state threads to let us extract (pipeline) parallelism.
We also identify other structures in state threads that are inherently parallel.
\nr{%
	Alternative version of the last sentence -- if we want to stress the different kinds of parallelism here already: %
	\emph{%
		We also identify other structures in state threads that naturally lead to data parallelism and task-level parallelism.
	}%
}
\subsection{Foundations}%
\label{sec:cat:foundations}%
\ourLang is a typed $\lambda$-calculus extended with state threads.
The details of the $\lambda$-calculus are not important, and almost any typed $\lambda$-calculus can be augmented with state threads to yield an implementation of \ourLang.
For our formal model of state threads presented in this section it is only relevant that the semantics of the $\lambda$-calculus can be interpreted in category-theoretic terms.

Let \catH be the category whose objects \obj{\catH} are the types in the $\lambda$-calculus and whose morphisms \morph{\catH} are the functions of the $\lambda$-calculus.
The category \catH is required to be \emph{cartesian closed}, which essentially means that for any types $a, b \in\obj{\catH}$, the product type $a\times b$ and the function type $a\rightarrow b$ exist, i.e.~$a\times b \in\obj{\catH}$ and $a\rightarrow b \in\obj{\catH}$.
Examples of cartesian closed categories are the categories of domains typically encountered in denotational semantics.

In more concrete terms, since most functional programming languages are fancy $\lambda$-calculi, \ourLang can be built on top of almost any functional language.
In the case of Haskell, for example, the category \catH is known as \emph{Hask}.%
\footnote{See \url{https://wiki.haskell.org/Hask}, although full \emph{Hask} is not cartesian closed, and may in fact not even be a category (cf.~\url{http://math.andrej.com/2016/08/06/hask-is-not-a-category/}).}
\subsection{State threads}%
\label{sec:cat:state:threads}%
In \ourLang, every state thread has its own private state that it operates on.
Hence, state threads and their respective states are both indexed by the same index set, henceforth denoted as $N$.
In practice, $N$ is typically finite, but it is generally sufficient to assume that $N$ is countable, i.e.~$N \cong \mathbb{N}$.

For the formal development of \ourLang in the present section, it is convenient to require not only that each state thread has its own state, but also that every state is of a distinct type.
Types are objects in the category \catH, leading to the following definition.
\begin{definition}[State objects, global state]
\label{thm:def:global:state}
Let $N$ be a countable index set.
For $n\in N$, let $s_n\in\obj{\catH}$ be pairwise distinct %
(i.e.~$s_n = s_m \Rightarrow n = m$).
\begin{enumerate}
	\item
	For $I \subseteq N$, define $s_I = \prod_{n\in I} s_n$.
	The $s_I$ are called state objects.
	\item
	The state objects $s_n$, for $n\in N$, are called fundamental.
	\item
	The state object $s_N = \prod_{n\in N} s_n$ is called the global state.
\end{enumerate}%
\end{definition}%
Note that $s_{\{n\}} = \prod_{m\in \{n\}} s_m = s_n$, $n\in N$, %
i.e.~the fundamental state objects are precisely the state objects $s_I$ for which $I\subseteq N$ has cardinality~$1$.
We also use the convention $s_{\emptyset} = ()$, i.e.~the unit type.

The requirement that the $s_n$ be pairwise distinct is not a restriction of \ourLang's programming model.
In Haskell, one can use the $\texttt{newtype}$ keyword to generate new and distinct types.
Typically, $\lambda$-calculi with less advanced type systems also offer ways of constructing new types in similar ways, e.g.~by suitably tagging types.

Having introduced state objects, we can now define \ourLang's state threads.
It is then readily seen that state threads form a subcategory of \catH, which we refer to as the \emph{category of sate threads}.
\begin{definition}[State thread]
\label{thm:def:state:thread}
Let $\{s_n\}_{n\in N}$ be the set of fundamental state objects.
A \emphdef{state thread} is a morphism $f \in \morph{\catH}$ such that
\begin{align}
	f : (a\times s_I) \rightarrow (b\times s_I) \,,
	\label{eq:def:state:thread}
\end{align}
where $I\subseteq N$. %
A \emphdef{fundamental} state thread is a state thread %
$f : (a\times s_n) \rightarrow (b\times s_n)$,
i.e.~a state thread for which $I= \{n\}$, $n\in N$, in Equation~\eqref{eq:def:state:thread}.
\end{definition}%
\begin{lemma}
\label{thm:lem:category:state:thread}
The following define the objects and morphisms of a subcategory \catS of \catH,
\begin{align}
	& \obj{\catS} = \left\{ a\times s_I \mid
		a\in\obj{\catH},
	 	I\subseteq N \right\} \,, \\
	& \morph{\catS} = \left\{ f : (a\times s_I) \rightarrow (b\times s_I) \mid
		f\in\morph{\catH},
		I\subseteq N \right\} \,.
\end{align}%
\begin{proof}
Clearly, $\id_{a\times s_I} \in \morph{\catS}$.
\catS inherits composition of morphisms from \catH.
Now, let $f,g \in \morph{\catS}$.
Whenever $g\circ f$ is defined in \catH, then $g\circ f \in\morph{\catS}$ follows directly by inspecting the signatures of $f$, $g$, and $g\circ f$.
\end{proof}%
\end{lemma}%
\begin{definition}[Category of state threads]
\label{thm:def:category:state:thread}
The category \catS from Lemma~\ref{thm:lem:category:state:thread} is called the \emphdef{category of state threads}.
\end{definition}%

The intuition is that the function that corresponds to the state thread $f : (a\times s_I) \rightarrow (b\times s_I)$ in the underlying $\lambda$-calculus only manipulates the part $s_I$ of the global state $s_N$, $I \subseteq N$.
The proof of Lemma~\ref{thm:lem:category:state:thread} relies on the fact that state threads %
$f : (a\times s_I) \rightarrow (b\times s_I)$ and %
$g : (b\times s_J) \rightarrow (c\times s_J)$ %
can be composed (in \catH or \catS) if and only if $I = J$.
(This observation relies on the pairwise distinctness of the $\{s_n\}_{n\in N}$.)
In the intuition just given, this means that $f$ and $g$ operate on the exact same part of the global state.
Without additional information about the structure of $f$ and $g$, an implementation of \ourLang is then forced to evaluate the composition $g \circ f$ sequentially.
However, an implementation can potentially exploit parallelism if $I \cap J = \emptyset$, i.e.~when $f$ and $g$ operate on disjoint parts of the global state.
The next section explains how \ourLang facilitates the composition of state threads $f : (a\times s_I) \rightarrow (b\times s_I)$ and %
$g : (b\times s_J) \rightarrow (c\times s_J)$ %
with arbitrary $I,J\subseteq N$.
\out{%
	\nr{%
		Here we have effectively motivated our analysis of the composition of state threads by our desire to find implicit parallelism.
		Therefore, we should really stress parallelism at the beginning of this section.
		This is indeed done now, but I wonder whether is is stressed enough in the first paragarph of the present section.
	}%
	}
\subsection{Composition of state threads}
\label{sec:composition:state:threads}
At a high level, \ourLang programs are composed of state threads, and compositions can ultimately be broken down into fundamental state threads.
From now on, we assume that the fundamental state threads that occur in a given \ourLang program are in 1-1 correspondence with the index set $N$.
The following definition introduces the symbol \calM to refer to the set of fundamental state threads in a program, i.e.~the \emph{state threads of interest}.
\begin{definition}[Fundamental state threads of interest]
\label{thm:def:interest:state:thread}
Let $N$ be an index set and let $\{s_n\}_{n\in N}$ be the (pairwise distinct) fundamental state objects, as in the previous section.
Let $\calM \subseteq \morph{\catS}$, %
and assume there is a bijective map $\phi_{\calM} : N \rightarrow \calM$ (i.e.~a 1-1 correspondence) such that %
\begin{align}
	\phi_{\calM}(n) : (a_n\times s_n) \rightarrow (b_n\times s_n) \,,
\end{align}
where $a_n,b_n\in\obj{\catH}$.
Then the elements of \calM are the \emphdef{fundamental state threads of interest}.
\end{definition}%
\ourLang handles state implicitly.
This motivates the organization of the fundamental state threads in \calM into a graph that hides the state objects but makes the possibility of composition explicit.

\begin{definition}[Multi-graph of fundamental state threads]
\label{thm:def:multi-graph}
Let \calM and $\phi_{\calM}$ as in Definition~\ref{thm:def:interest:state:thread}.
The directed (multi-)graph $\Delta_{\calM}$ has the following vertices ($V$) and edges ($E$),
\begin{align}
	& V(\Delta_{\calM}) = \{ a_n, b_n \mid
							 \phi_{\calM}(n) : (a_n\times s_n) \rightarrow (b_n\times s_n), n\in N \} \,, \\
 	& E(\Delta_{\calM}) = N \,,
\end{align}%
and the maps $\Src,\Tgt : E(\Delta_{\calM}) \rightarrow V(\Delta_{\calM})$ are defined as follows,
\begin{align}
	\src{n} = a_n, \text{ if } \phi_{\calM}(n) : (a_n\times s_n) \rightarrow (b_n\times s_n) \,,
	\label{eq:multi-graph:src} \\
	\tgt{n} = b_n, \text{ if } \phi_{\calM}(n) : (a_n\times s_n) \rightarrow (b_n\times s_n) \,.
	\label{eq:multi-graph:tgt}
\end{align}%
\end{definition}%

Note that directed multi-graphs are also referred to as \emph{quivers} in the literature.
Also note that in the light of Equations~\eqref{eq:multi-graph:src} and~\eqref{eq:multi-graph:tgt}, the signature of $\phi_{\calM}(n)$ can be written without referring to the (arbitrary) objects $a_n$, $b_n$:
\begin{align}
	\phi_{\calM}(n) : (\src{n}\times s_n) \rightarrow (\tgt{n}\times s_n) \,.
\end{align}%

\begin{figure*}
\begin{minipage}[b]{0.4\linewidth}
\begin{center}
\begin{tikzpicture}[
	every node/.style = {align=center}]
	\node (a) {\tiny $a$};
	\node[right =12mm of a] (b) {\tiny $b$};
	\node[above right=10mm of b] (c) {\tiny $c$};
	\node[below right=10mm of b] (e) {\tiny $e$};
	\node[below right=10mm of c] (d) {\tiny $d$};

	\path[->, looseness=1.0] (a) edge[out=30,in=150] node[midway, above] {\tiny $1$} (b);
	\path[->, looseness=1.0] (a) edge[out=-30,in=-150] node[midway, below] {\tiny $2$} (b);

	\path[->] (c) edge node[midway, above] {\tiny $3$} (b);
	\path[->] (b) edge node[midway, below] {\tiny $4$} (e);
	\path[->] (c) edge node[midway, above] {\tiny $5$} (d);

	\path[->] (e) edge node[midway, above] {\tiny $6$} (d);
	\path[->, looseness=1.0] (d) edge[out=-70,in=-10] node[midway, below] {\tiny $7$} (e);
\end{tikzpicture}%
	\captionof{figure}{Example of a multi-graph $\Delta_{\calM}$ of fundamental state threads for $N=\{1,\dots,7\}$.}%
	\label{fig:example:multi-graph}%
\end{center}%
\end{minipage}%
\hspace{8mm}%
\begin{minipage}[b]{0.46\linewidth}
\begin{center}
\begin{tikzpicture}[
	every node/.style = {align=center}]
	\node (a) {\tiny $a\times s_N$};
	\node[right =12mm of a] (b) {\tiny $b\times s_N$};
	\node[above right=10mm of b] (c) {\tiny $c\times s_N$};
	\node[below right=10mm of b] (e) {\tiny $e\times s_N$};
	\node[below right=10mm of c] (d) {\tiny $d\times s_N$};

	\path[->, looseness=1.0] (a) edge[out=30,in=150] node[midway, above] {\tiny $\phi_{\calM}^{*}(1)$} (b);
	\path[->, looseness=1.0] (a) edge[out=-30,in=-150] node[midway, below] {\tiny $\phi_{\calM}^{*}(2)$} (b);

	\path[->] (c) edge node[midway, above] {\tiny $\phi_{\calM}^{*}(3)$} (b);
	\path[->] (b) edge node[midway, below] {\tiny $\phi_{\calM}^{*}(4)$} (e);
	\path[->] (c) edge node[midway, above] {\tiny $\phi_{\calM}^{*}(5)$} (d);

	\path[->] (e) edge node[midway, above] {\tiny $\phi_{\calM}^{*}(6)$} (d);
	\path[->, looseness=1.0] (d) edge[out=-70,in=-10] node[midway, below] {\tiny $\phi_{\calM}^{*}(7)$} (e);
\end{tikzpicture}%
	\captionof{figure}{The corresponding subcategory $\catC_{\calM}$ in \catS.}
	\label{fig:example:subcategory}%
\end{center}%
\end{minipage}%
\end{figure*}%
Figure~\ref{fig:example:multi-graph} gives an example of a multi-graph $\Delta_{\calM}$ for seven fundamental state threads.
Note how the state objects do not appear explicitly; they are, however, implicit in the naming of the edges.
By contrast, composition of the state threads $\phi_{\calM}(m)$ and $\phi_{\calM}(n)$ is explicitly suggested whenever either $\tgt{n} = \src{m}$ or $\tgt{m} = \src{n}$.

The composition of state threads is natural in \catS, and we would like to use this composition also for the state threads in \calM.
To facilitate this, we now construct a suitable embedding of the graph $\Delta_{\calM}$ into the category \catS.
Our language is deliberately imprecise here to convey the right intuition.
More correctly, we should speak of embedding $\Delta_{\calM}$ into the graph underlying $\catS$.
Said yet another way, we are looking for a way to identify the free category over $\Delta_{\calM}$ inside $\catS$;
and this is precisely what is achieved by the remaining definitions and lemma in the present section.
\begin{definition}[Free category over a graph]%
\label{thm:def:free:category}%
The \emphdef{free category} $\freeCat{\Delta_{\calM}}$ over $\Delta_{\calM}$ is the category whose objects are the vertices of $\Delta_{\calM}$ and whose morphisms are precisely the paths in $\Delta_{\calM}$, i.e.
\begin{align}
	& \obj{\freeCat{\Delta_{\calM}}} = V(\Delta_{\calM}), \\
	& \morph{\freeCat{\Delta_{\calM}}} = \left\{
		n_k n_{k-1} \dots n_2 n_1 \mid
		k \in \mathbb{N},
		n_i\in N,
		\tgt{n_i} = \src{n_{i+1}} \text{ for } 1\le i \le k-1 \right\} \nonumber \\
	& \phantom{\morph{\freeCat{\Delta_{\calM}}} = }\text{ }
	  \cup \left\{\epsilon_v \mid v\in V(\Delta_{\calM})\right\}.
\end{align}%
\end{definition}%
This definition of the free category over a graph is completely standard.
Note that we take a separate copy of the empty path $\epsilon$ for each vertex $v$ of $\Delta_{\calM}$.
In categorical terms, $\epsilon_v : v\rightarrow v$ is the identity morphism at the object $v$.
The morphisms $\morph{\freeCat{\Delta_{\calM}}}$ can be thought of as words over the alphabet $N$.
In the following, we adopt this point of view.
Note that there is then a separate copy of the empty word for each vertex $v$ of $\Delta_{\calM}$.

By the universal property of the product, any state thread $f : (a\times s_I) \rightarrow (b\times s_I)$, with $I\subseteq N$, has a natural (and unique) extension to a state thread that operates on the global state $s_N$.
\begin{definition}[Extension of state threads]
\label{thm:def:state:thread:extension}%
Let $f : (a\times s_I) \rightarrow (b\times s_I)$ be a state thread.
The state thread $f^{*} : (a\times s_N) \rightarrow (b\times s_N)$ is obtained from $f$ by extending $f$ with the identity on $s_{N\setminus I}$.
\end{definition}%
Using this extension of state threads to the global state $s_N$, we can finally define the functor that identifies the graph $\Delta_{\calM}$ inside the category of state threads.
\begin{definition}
\label{thm:def:Phi:functor}%
The functor $\Phi_{\calM} : \freeCat{\Delta_{\calM}} \rightarrow \catS$ is defined by %
$\Phi_{\calM}(v) = v \times s_N$ for objects $v\in\obj{\freeCat{\Delta_{\calM}}}$ and by
\begin{align}
	& \Phi_{\calM}(\epsilon_v) = \id_{v\times s_N}
		\label{eq:def:Phi:functor:1} \,, \\
	& \Phi_{\calM}(n_k n_{k-1} \dots n_2 n_1) =
		\phi_{\calM}^{*}(n_k) \circ \phi_{\calM}^{*}(n_{k-1}) \circ \dots
		\circ \phi_{\calM}^{*}(n_2) \circ \phi_{\calM}^{*}(n_1)
		\label{eq:def:Phi:functor:2}
\end{align}%
for morphisms in $\morph{\freeCat{\Delta_{\calM}}}$.
The composition on the right-hand side of Equation~\eqref{eq:def:Phi:functor:2} is the composition in \catS (which is the same as in \catH).
\end{definition}%

Based on Equations~\eqref{eq:def:Phi:functor:1} and~\eqref{eq:def:Phi:functor:2}, the functor properties are readily verified for $\Phi_{\calM}$.
More interestingly, $\Phi_{\calM}$ picks out a subcategory in \catS.
\begin{lemma}
\label{thm:lem:image:Phi}%
The image of $\,\Phi_{\calM}$ forms a subcategory of \catS.
\begin{proof}
Straightforward.
The only subtle aspect is that for two words $w_1,w_2\in\morph{\freeCat{\Delta_{\calM}}}$
such that the composition $\Phi_{\calM}(w_2) \circ \Phi_{\calM}(w_1)$ is in \catS, one must show that  $\Phi_{\calM}(w_2) \circ \Phi_{\calM}(w_1)$ is in the image of $\Phi_{\calM}$.
Now, if $\Phi_{\calM}(w_2)$ and $\Phi_{\calM}(w_1)$ can be composed in \catS, then $\tgt{w_1} = \src{w_2}$, %
with natural extensions of $\Src$, $\Tgt$ from letters in $N$ to words in $\morph{\freeCat{\Delta_{\calM}}}$.
But then, $w_2w_1 \in \morph{\freeCat{\Delta_{\calM}}}$, and hence $\Phi_{\calM}(w_2) \circ \Phi_{\calM}(w_1) = \Phi_{\calM}(w_2w_1)$ is in the image of $\,\Phi_{\calM}$.
\end{proof}
\end{lemma}%
\begin{definition}[Image of $\,\Phi_{\calM}$]
\label{thm:def:image:Phi}%
The subcategory of \catS that is the image of $\Phi_{\calM}$ is denoted as $\catC_{\calM}$.
\end{definition}%

In summary, by extending the state threads of interest to operate on the global state $s_N$, it has become possible to compose state threads %
$f^{*} : (a\times s_N) \rightarrow (b\times s_N)$ and %
$g^{*} : (b\times s_N) \rightarrow (c\times s_N)$ %
even if the original state threads $f$, $g$ operate on disjoint parts $s_I$ and $s_J$ of the global state.
At the same time, the information that the extended state thread $f^{*}$ leaves the state $s_{N\setminus I}$ unchanged is retained by the fact $f^{*} = \Phi_{\calM}(w)$, for some $w\in\morph{\freeCat{\Delta_{\calM}}}$.
In fact, the letters from $N$ that occur in $w$ are precisely the elements of the subset $I\subseteq N$.
An analogous statement holds for $g^{*}$.

Moreover, we have identified the subcategory $\catC_{\calM}$ of \catS that is generated by the state threads of interest in \calM.
Figure~\ref{fig:example:subcategory} visualizes how $\catC_{\calM}$ is related to the multi-graph $\Delta_{\calM}$ from Figure~\ref{fig:example:multi-graph}.
\subsection{The \smap functor}%
\label{sec:smap}%
\nr{%
	It is not fair to speak of a ``functor for state threads'' since we do not define \smap on the category of state threads (i.e.~\catS), but only on a subcategory.
	On another note, I also would not start this section speaking of \smap as a ``higher-order'' function if we have just said in the section title that it is a functor.
	(I know that in the terminology of Haskell, both are not mutually exclusive, but I would try to reduce the potential for confusion here.)
}%

The functor $\Phi_{\calM}$ from Definition~\ref{thm:def:Phi:functor} is not the only way of identifying $\freeCat{\Delta_{\calM}}$ as a subcategory in $\catS$.
Recall that the objects of $\freeCat{\Delta_{\calM}}$ are the vertices of the multi-graph $\Delta_{\calM}$, which in turn are objects of \catH, i.e.~types in the $\lambda$-calculus on that \ourLang is based.
An alternative way of identifying $\freeCat{\Delta_{\calM}}$ in $\catS$ is obtained by mapping the objects of $\freeCat{\Delta_{\calM}}$ to list types.
By making this precise we will naturally be led to the \smap functor, i.e.~the functor that generalizes \textit{map} to state threads.
\begin{definition}%
\label{thm:def:psi:map}%
Let \calM be the set of state threads of interest, %
and let $\phi_{\calM} : N \rightarrow \calM$ be the corresponding bijective map.
For each $n\in N$, recursively define a state thread $\psi_{\calM}(n)$ as follows,
\begin{align}
	& \psi_{\calM}(n) : ([\src{n}] \times s_n) \rightarrow ([\tgt{n}] \times s_n) \\
	& \psi_{\calM}(n) \: ([], \sigma) = ([], \sigma)
		\label{eq:def:psi:map:1} \\
	& \psi_{\calM}(n) \: (x : xs, \sigma) = \texttt{let }
		(y, \sigma') = \phi_{\calM}(n) (x, \sigma) \nonumber \\
	& \phantom{\psi_{\calM}(n) \: (x:xs, \sigma) = \texttt{let }}
		(ys, \sigma'') = \psi_{\calM}(n) \: (xs, \sigma') \nonumber \\
	& \phantom{\psi_{\calM}(n) \: (x:xs, \sigma) = \texttt{}} \texttt{in }
		(y:ys, \sigma'') \,,
		\label{eq:def:psi:map:2}
\end{align}
where $\Src$ and $\Tgt$ are the maps defining the multi-graph $\Delta_{\calM}$ from Definition~\ref{thm:def:multi-graph}.
\end{definition}%
\begin{definition}
\label{thm:def:Psi:functor}%
The functor $\Psi_{\calM} : \freeCat{\Delta_{\calM}} \rightarrow \catS$ is defined by %
$\Psi_{\calM}(v) = [v] \times s_N$ for objects $v\in\obj{\freeCat{\Delta_{\calM}}}$ and by
\begin{align}
	& \Psi_{\calM}(\epsilon_v) = \id_{[v]\times s_N}
		\label{eq:def:Psi:functor:1} \,, \\
	& \Psi_{\calM}(n_k n_{k-1} \dots n_2 n_1) =
		\psi_{\calM}^{*}(n_k) \circ \psi_{\calM}^{*}(n_{k-1}) \circ \dots
		\circ \psi_{\calM}^{*}(n_2) \circ \psi_{\calM}^{*}(n_1)
		\label{eq:def:Psi:functor:2}
\end{align}%
for morphisms in $\morph{\freeCat{\Delta_{\calM}}}$.
\end{definition}%

Exactly as in Lemma~\ref{thm:lem:image:Phi} one verifies that the image of $\Psi_{\calM}$ is a subcategory of \catS.
\begin{definition}[Image of $\,\Psi_{\calM}$]
\label{thm:def:image:Psi}%
The subcategory of \catS that is the image of $\Psi_{\calM}$ is denoted as $\catC_{\calM}^{[]}$.
\end{definition}%

The \smap functor will be defined to mediate between the categories $\catC_{\calM}$ and $\catC_{\calM}^{[]}$.
This means that, analogously to the \textit{map} functor, \smap takes a state thread with signature $(a\times s_N) \rightarrow (b\times s_N)$ and returns a state thread with signature $([a]\times s_N) \rightarrow ([b]\times s_N)$.
Additionally, if the argument of \smap is composed of multiple fundamental state threads, \smap implements the appropriate plumbing of state in the resulting state thread $([a]\times s_N) \rightarrow ([b]\times s_N)$.

Before we can define \smap, we need a lemma that states that, under certain conditions, the value of the functor $\Phi_{\calM}$ fully determines $\Psi_{\calM}$.
\begin{lemma}%
\label{thm:lem:Phi:determines:Psi}%
Let $w_1,w_2\in\morph{\freeCat{\Delta_{\calM}}}$ be such that no letter of $N$ occurs more than once in either $w_1$ or $w_2$.
Then,
\begin{align}
	\Phi_{\calM}(w_1) = \Phi_{\calM}(w_2) \:\Rightarrow\: \Psi_{\calM}(w_1) = \Psi_{\calM}(w_2) \,.
\end{align}%
\begin{proof}
The proof appears in Section~\ref{sec:smap:well:defined}.
It relies on an algebraic manipulation that is known as \emph{let floating} in the context of functional language compilers~\cite{PeytonJones:1996:LMB:232627.232630}.
\nr{%
	I am leaving the remark about \emph{let floating} here for readers who do not read the detailed proof in Section~\ref{sec:smap:well:defined}.
}%
\out{%
	\nr{%
		In case we do in the end decide to omit the section with the detailed proof,
		the following text can be used here instead (as a short proof):
	}%
	One shows by induction on the length of, say, $w_1$ that $\Phi_{\calM}(w_1)$ fully determines $\Psi_{\calM}(w_1)$.
	The inductive step requires a manipulation that is known as \emph{let floating} in the context of functional language compilers~\cite{PeytonJones:1996:LMB:232627.232630}.
}
\end{proof}%
\end{lemma}%
\begin{theorem}[and definition of \smap]
\label{thm:thm:smap}%
If the multi-graph $\Delta_{\calM}$ has no cycles, %
then the following define a functor $\smap : \catC_{\calM} \rightarrow \catC_{\calM}^{[]}$,
\begin{align}
	& \smap\left(v\times s_N\right) = [v]\times s_N \,,
		\text{ for } v\in\obj{\freeCat{\Delta_{\calM}}}   \label{eq:smap:definition:objects} \\
	& \smap\left(\Phi_{\calM}(w)\right) = \Psi_{\calM}(w) \,,
		\text{ for } w\in\morph{\freeCat{\Delta_{\calM}}} \label{eq:smap:definition:morphisms}
	\,.
\end{align}%
\begin{proof}
Since $\Delta_{\calM}$ has no cycles, no letter from $N$ can occur more than once in any $w\in\morph{\freeCat{\Delta_{\calM}}}$.
Hence, Lemma~\ref{thm:lem:Phi:determines:Psi} guarantees that \smap is well-defined.

Verifying the functor properties is mechanical.
Let $w_1, w_2 \in\morph{\freeCat{\Delta_{\calM}}}$, and assume %
$w_1 = n_k\dots n_1$, %
$w_2 = m_l\dots m_1$, %
with $m_l\dots m_1, n_k\dots n_1 \in N$.
Then,
\begin{align}
	\smap\left(\Phi_{\calM}(w_2) \circ \Phi_{\calM}(w_1)\right)
		& = \smap\left(
				\phi_{\calM}^{*}(m_l) \circ \dots
				\circ \phi_{\calM}^{*}(m_1)
				\circ \phi_{\calM}^{*}(n_k) \circ \dots
				\circ \phi_{\calM}^{*}(n_1)
				\right) \\
		& = \smap\left(\Phi_{\calM}(w_2 w_1)\right) \\
		& = \Psi_{\calM}(w_2 w_1) \\
		& = \Psi_{\calM}(m_l \dots m_1 n_k \dots n_1) \\
		& = \psi_{\calM}^{*}(m_l) \circ \dots \circ \psi_{\calM}^{*}(m_1)
				\circ \psi_{\calM}^{*}(n_k) \circ \dots
				\circ \psi_{\calM}^{*}(n_1) \\
		& = \Psi_{\calM}(w_2) \circ \Psi_{\calM}(w_1) \\
		& = \smap\left(\Phi_{\calM}(w_2)\right) \circ
				\smap\left(\Phi_{\calM}(w_1)\right)
	\,.
\end{align}
\end{proof}%
\end{theorem}%
\subsection{Extracting parallelism from the structure of state threads}%
\label{sec:parallel}%
Having defined state threads in \ourLang and the \smap functor, we now investigate opportunities for extracting parallelism based on the structure of state threads.
We show that pipeline parallelism arises naturally from \smap, and we identify structures that exhibit data and task-level parallelism.
%
%
\subsubsection{Pipeline parallelism}%
\label{sec:pipeline:parallelism}%
The \smap functor is defined in terms of $\Psi_{\calM}$, for which Equation~\eqref{eq:def:Psi:functor:2} suggests a very sequential implementation:
to evaluate $\Psi_{\calM}(n_k\dots n_1)$ on an input $(\xs,\sigma)\in[a]\times s_N$, one should first apply $\psi_{\calM}^{*}(n_1)$, then $\psi_{\calM}^{*}(n_2)$, and so on.
By Definition~\ref{thm:def:psi:map}, this means that $\phi_{\calM}(n_1)$ is first applied to every element of the list \xs before $\phi_{\calM}(n_2)$ is applied etc.
To obtain pipeline parallelism, this order must be relaxed.

How this can be done is illustrated in Figure~\ref{fig:smap-pipe} for $k=2$.
The top diagram in Figure~\ref{fig:smap-pipe} is a graphical representation of Equation~\eqref{eq:def:Psi:functor:2} applied to the argument %
$([x_1,\dots,x_l], (\sigma_{n_1},\sigma_{n_2},\widetilde{\sigma}))\in[a]\times s_N$.
Red and blue arrows indicate which components of this argument are modified by applications of $\phi_{\calM}(n_1)$ and $\phi_{\calM}(n_2)$ respectively.
Note that each application of $\phi_{\calM}(n_1)$ and $\phi_{\calM}(n_2)$ modifies two components, and hence there are two arrows in every column of the top diagram.
The bottom diagram in Figure~\ref{fig:smap-pipe} can be thought of as a squeezed version of the top diagram.
In all but the first and the last column there are now four arrows: %
one pair of red arrows and one pair of blue arrows.
This indicates that $\phi_{\calM}(n_1)$ and $\phi_{\calM}(n_2)$ can be evaluated in parallel, yielding pipeline parallelism.
Note that while the top diagram has $2l$ columns, the bottom one only has $l\!+\!1$.
The data flowing through the pipeline are the elements of the lists %
$[x_1,\dots,x_l]$, $[y_1,\dots,y_l]$, and $[z_1,\dots,z_l]$.

Squeezing the top diagram of Figure~\ref{fig:smap-pipe} into the bottom diagram is possible since $\phi_{\calM}(n_1)$ and $\phi_{\calM}(n_2)$ operate on different fundamental state objects, i.e.~$n_1\ne n_2$.
That $n_1\ne n_2$ follows from the fact that the multi-graph $\Delta_{\calM}$ is acyclic, which was required to ensure that \smap is well-defined by Equation~\eqref{eq:smap:definition:morphisms}.
When $\Delta_{\calM}$ has cycles, pipeline parallelism can still be exploited in evaluating $\Psi_{\calM}(n_k\dots n_1)$ provided the $n_1,\dots,n_k\in N$ are pairwise distinct.
\out{%
	\nr{%
		Talking about $\Psi_{\calM}$ (and not \smap) in the previous sentence is a neat way of explaining the pipeline parallelism that is present for cyclic $\Delta_{\calM}$ without getting into discussions of what it means that \smap is not well-defined.
	}%
}%
More generally, for $w_1,w_2,w_2\in\morph{\freeCat{\Delta_{\calM}}}$ such that only $w_2$ contains multiple occurrences of the same letter in $N$, the functor property, i.e.~%
$\Psi_{\calM}(w_3 w_2 w_1) = \Psi_{\calM}(w_3) \circ \Psi_{\calM}(w_2) \circ \Psi_{\calM}(w_1)$, %
can be used to still exploit the parallelism in $\Psi_{\calM}(w_1)$ and $\Psi_{\calM}(w_3)$.
\newcommand{\redarrow}{\ensuremath{\textcolor{red}{\xrightarrow{\mathmakebox[4mm]{}}}}\xspace}
\newcommand{\bluearrow}{\ensuremath{\textcolor{blue}{\xrightarrow{\mathmakebox[4mm]{}}}}\xspace}
\begin{figure*}
{\footnotesize
\begin{minipage}[c]{\textwidth}
\begin{alignat*}{18}
\multicolumn{18}{l}{
	$\smap\left(\Phi_{\calM}(n_2 n_1)\right) \: ([x_1,\dots,x_l], (\sigma_{n_1},\sigma_{n_2},\widetilde{\sigma})) =
	 \Psi_{\calM}(n_2 n_1) \: ([x_1,\dots,x_l], (\sigma_{n_1},\sigma_{n_2},\widetilde{\sigma})) = $
}\\[1mm]
& ([x_1,
& \redarrow
& ([y_1,
&
& ([y_1,
&
& ([y_1,
&
& ([y_1,
& \bluearrow
& ([z_1,
&
& ([z_1,
&
& ([z_1,
&
& ([z_1,
\\
& \phantom{([} x_2,
&
& \phantom{([} x_2,
& \redarrow
& \phantom{([} y_2,
& \phantom{\redarrow}
& \phantom{([} y_2,
&
& \phantom{([} y_2,
&
& \phantom{([} y_2,
& \multicolumn{1}{c}{$\cdots$}
& \phantom{([} z_2,
&
& \phantom{([} z_2,
&
& \phantom{([} z_2,
\\
& \phantom{([} \vdots
&
& \phantom{([} \vdots
&
& \phantom{([} \vdots
& \multicolumn{1}{c}{$\ddots$}
& \phantom{([} \vdots
&
& \phantom{([} \vdots
&
& \phantom{([} \vdots
& \multicolumn{1}{c}{$\ddots$}
& \phantom{([} \vdots
&
& \phantom{([} \vdots
& \out{arrow for yl}
& \phantom{([}\vdots
\\
& \phantom{([} x_{l-1},
&
& \phantom{([} x_{l-1},
&
& \phantom{([} x_{l-1},
& \multicolumn{1}{c}{$\cdots$}
& \phantom{([} y_{l-1},
& \phantom{\redarrow}
& \phantom{([} y_{l-1},
&
& \phantom{([} y_{l-1},
&
& \phantom{([} y_{l-1},
& \bluearrow
& \phantom{([} z_{l-1},
&
& \phantom{([} z_{l-1},
\\
& \phantom{([} x_l],
&
& \phantom{([} x_l],
&
& \phantom{([} x_l],
&
& \phantom{([} x_l],
& \redarrow
& \phantom{([} y_l],
&
& \phantom{([} y_l],
& \phantom{\bluearrow}
& \phantom{([} y_l],
&
& \phantom{([} y_l],
& \bluearrow
& \phantom{([} z_l],
\\
& \phantom{(} (\sigma_{n_1},
& \redarrow
& \phantom{(} (\sigma_{n_1}^{(1)}
& \redarrow
& \phantom{(} (\sigma_{n_1}^{(2)}
& \multicolumn{1}{c}{$\cdots$}
& \phantom{(} (\sigma_{n_1}^{(l-1)}
& \redarrow
& \phantom{(} (\sigma_{n_1}^{(l)},
&
& \phantom{(} (\sigma_{n_1}^{(l)},
&
& \phantom{(} (\sigma_{n_1}^{(l)},
&
& \phantom{(} (\sigma_{n_1}^{(l)},
&
& \phantom{(} (\sigma_{n_1}^{(l)},
\\
& \phantom{((}\sigma_{n_2},
&
& \phantom{((}\sigma_{n_2},
&
& \phantom{((}\sigma_{n_2},
&
& \phantom{((}\sigma_{n_2},
&
& \phantom{((}\sigma_{n_2},
& \bluearrow
& \phantom{((} \sigma_{n_2}^{(1)},
& \multicolumn{1}{c}{$\cdots$}
& \phantom{((} \sigma_{n_2}^{(l-2)},
& \bluearrow
& \phantom{((} \sigma_{n_2}^{(l-1)},
& \bluearrow
& \phantom{((} \sigma_{n_2}^{(l)},
\\
& \phantom{((}\widetilde{\sigma}))
&
& \phantom{((}\widetilde{\sigma}))
&
& \phantom{((}\widetilde{\sigma}))
&
& \phantom{((}\widetilde{\sigma}))
&
& \phantom{((}\widetilde{\sigma}))
&
& \phantom{((} \widetilde{\sigma}))
&
& \phantom{((}\widetilde{\sigma}))
&
& \phantom{((}\widetilde{\sigma}))
&
& \phantom{((}\widetilde{\sigma}))
\end{alignat*}%
\end{minipage}
\begin{center}
$=$\\
\end{center}%
\begin{minipage}[c]{\textwidth}
\begin{alignat*}{12}
& ([x_1,
& \redarrow
& ([y_1,
& \bluearrow
& ([z_1,
&
& ([z_1,
&
& ([z_1,
&
& ([z_1,
\\
& \phantom{([} x_2,
&
& \phantom{([} x_2,
& \redarrow
& \phantom{([} y_2,
& \multicolumn{1}{c}{$\cdots$}
& \phantom{([} z_2,
&
& \phantom{([} z_2,
&
& \phantom{([} z_2,
\\
& \phantom{([} \vdots
&
& \phantom{([} \vdots
&
& \phantom{([} \vdots
& \multicolumn{1}{c}{$\ddots$}
& \phantom{([} \vdots
&
& \phantom{([} \vdots
&
& \phantom{([} \vdots
\\
& \phantom{([} x_{l-1},
&
& \phantom{([} x_{l-1},
&
& \phantom{([} x_{l-1},
& \multicolumn{1}{c}{$\cdots$}
& \phantom{([} y_{l-1},
& \bluearrow
& \phantom{([} z_{l-1},
&
& \phantom{([} z_{l-1},
\\
& \phantom{([} x_l],
&
& \phantom{([} x_l],
&
& \phantom{([} x_l],
& \phantom{\redarrow}
& \phantom{([} x_l],
& \redarrow
& \phantom{([} y_l],
& \bluearrow
& \phantom{([} z_l],
\\
& \phantom{(} (\sigma_{n_1},
& \redarrow
& \phantom{(} (\sigma_{n_1}^{(1)},
& \redarrow
& \phantom{(} (\sigma_{n_1}^{(2)},
& \multicolumn{1}{c}{$\cdots$}
& \phantom{(} (\sigma_{n_1}^{(l-1)},
& \redarrow
& \phantom{(} (\sigma_{n_1}^{(l)},
&
& \phantom{(} (\sigma_{n_1}^{(l)},
\\
& \phantom{([} \sigma_{n_2},
&
& \phantom{([} \sigma_{n_2},
& \bluearrow
& \phantom{([} \sigma_{n_2}^{(1)},
& \multicolumn{1}{c}{$\cdots$}
& \phantom{([} \sigma_{n_2}^{(l-2)},
& \bluearrow
& \phantom{([} \sigma_{n_2}^{(l-1)},
& \bluearrow
& \phantom{([} \sigma_{n_2}^{(l)},
\\
& \phantom{([} \widetilde{\sigma}))
&
& \phantom{([} \widetilde{\sigma}))
&
& \phantom{([} \widetilde{\sigma}))
&
& \phantom{([} \widetilde{\sigma}))
&
& \phantom{([} \widetilde{\sigma}))
&
& \phantom{([} \widetilde{\sigma}))
\end{alignat*}
\end{minipage}
}
\caption{%
	Graphical representation of the \smap functor.
	Red arrows indicate applications of $\phi_{\calM}(n_1)$, %
	blue arrows indicate applications of $\phi_{\calM}(n_2)$.
	The top diagram is a direct representation based on the definition of $\Psi_{\calM}$ in Equation~\eqref{eq:def:Psi:functor:2}.
	The equivalent diagram on the bottom exhibits the inherent pipeline parallelism of \smap.
}%
\label{fig:smap-pipe}%
\end{figure*}%
\subsubsection{Data parallelism}%
\label{sec:data:parallelism}%
When fundamental state threads have certain additional structure, \smap reduces to \map, enabling the exploitation of data parallelism.
In the following, two structures for which this is possible are presented.

First, consider a morphism in \catH of the form %
$f : a \times s_n \rightarrow b$, %
which uses the state object $s_n$ in a read-only fashion (similar to Haskell's \texttt{Reader} type).
By the universal property of the product, we can extend $f$ to a state thread, i.e.~to a morphism $\tilde{f}$ in \catS by setting $\tilde{f}(x,\sigma) = (f(x,\sigma), \sigma)$ for $x\in a$ and $\sigma\in s_{n}$.
The left pane of Figure~\ref{fig:universal:product} gives the corresponding universal diagram.
\begin{figure}
\begin{minipage}{0.5\linewidth}
\begin{center}
	{\footnotesize
	\begin{tikzpicture}
		\matrix (m) [matrix of math nodes,row sep=3em,column sep=4em,minimum width=2em]
		{
			  & a\times s_n &     \\
			b & b\times s_n & s_n \\};
		\path[-stealth]
		(m-1-2) edge node [left] {$f$} (m-2-1)
				edge node [right] {$\pi_2$} (m-2-3)
		(m-2-2) edge node [above] {$\pi_1$} (m-2-1)
				edge node [above] {$\pi_2$} (m-2-3)
		(m-1-2) edge [dashed] node [right] {$\tilde{f}$} (m-2-2) ;
	\end{tikzpicture}%
	}
\end{center}%
\end{minipage}%
\begin{minipage}{0.5\linewidth}
\begin{center}
	{\footnotesize
	\begin{tikzpicture}
		\matrix (m) [matrix of math nodes,row sep=3em,column sep=4em,minimum width=2em]
		{
			  & a\times s_n &     \\
			b & b\times s_n & s_n \\};
		\path[-stealth]
		(m-1-2) edge node [left] {$g\circ\pi_1$} (m-2-1)
				edge node [right] {$h\circ\pi_2$} (m-2-3)
		(m-2-2) edge node [above] {$\pi_1$} (m-2-1)
				edge node [above] {$\pi_2$} (m-2-3)
		(m-1-2) edge [dashed] node [right] {$g\!\times\!h$} (m-2-2) ;
	\end{tikzpicture}%
	}
\end{center}%
\end{minipage}%
	\caption{%
		Universal diagrams for the product $b\times s_n$ with the natural projections $\pi_1$ and $\pi_2$.
	}%
	\label{fig:universal:product}%
\end{figure}%
If, in the notation introduced in Section~\ref{sec:cat:state:threads}, $\tilde{f} \in \calM$, then $\tilde{f} = \phi_{\calM}(n)$, and hence $\tilde{f}^{*} = \Phi_{\calM}(n)$.
Evaluating $\smap(\tilde{f}^{*})$ requires $\psi_{\calM}(n)$, whose defining Equation~\eqref{eq:def:psi:map:2} reduces to
\begin{align}
	& \psi_{\calM}(n) \: (\xs, \sigma) = \texttt{let }
		\ys = \map \: (x \mapsto f(x, \sigma)) \: xs \nonumber \\
	& \phantom{\psi_{\calM}(n) \: (\xs, \sigma) = \texttt{}} \texttt{in }
		(\ys, \sigma) \,,
\end{align}
and data parallelism can be exploited in evaluating \map.

The second instance of data parallelism arises if a fundamental state thread $(a \times s_n) \rightarrow (b \times s_n)$ operates independently on $a$ and $s_n$.
To see this, let %
$g : a \rightarrow b$ and %
$h : s_n \rightarrow s_n$ be morphisms in \catH.
Again, the universal property of the product can be used to construct a fundametal state thread $g\!\times\!h = \phi_{\calM}(n)$, as in the right pane of Figure~\ref{fig:universal:product}.
Alternatively, $g\!\times\!h$ is characterized by $(g\times h)(x, \sigma) = (g(x), h(\sigma))$.
Now, Equation~\eqref{eq:def:psi:map:2} for the corresponding $\psi_{\calM}(n)$ reduces to
\begin{align}
	& \psi_{\calM}(n) \: (\xs, \sigma) = \texttt{let }
		\ys = \map \: g \: \xs \nonumber \\
	& \phantom{\psi_{\calM}(n) \: (\xs, \sigma) = \texttt{let }}
		\sigma' = (\underbrace{h \circ \cdots \circ h}_{\textit{length}(\xs) \text{ times}}) \: \sigma \nonumber \\
	& \phantom{\psi_{\calM}(n) \: (\xs, \sigma) = \texttt{}} \texttt{in }
		(\ys, \sigma') \,.
		\label{eq:psi:independent:data:parallel}
\end{align}
Again, data parallelism can be exploited in evaluating \map.

Observe that while $\map \: g \: \xs$ in Equation~\eqref{eq:psi:independent:data:parallel} is data-parallel, the values of $\ys$ and $\sigma'$ can be computed in parallel too, which is an instance of task-level parallelism.
%
%
\subsubsection{Task-level parallelism}%
\label{sec:task:level:parallelism}%
The simplest case of task-level parallelism occurs if a state thread $h : (a \times b \times s_I \times s_J) \rightarrow (c \times d \times s_I \times s_J)$ with $I,J\subseteq N$ and $I\cap J = \emptyset$ decomposes into $f : (a \times s_I) \rightarrow (c \times s_I)$ and $g : (b\times s_J) \rightarrow (d\times s_J)$, %
i.e.~$h = f\!\times\!g$ using the same construction and notation as in the right diagram in Figure~\ref{fig:universal:product}.
Then, $h$ can be evaluated by executing $f$ and $g$ in parallel.
Here \smap is not required to arrive at parallelism.

A more interesting case occurs when the underlying category \catH has coproducts, i.e.,~if for any $a, b \in \obj{\catH}$, there exists an object $a + b \in \obj{\catH}$ and natural injections $\inl : a \rightarrow a+b$, $\inr : b \rightarrow a+b$.
Then, consider the following fundamental state threads, together with their extensions to $s_N$,
\out{%
with coproducts which implement conditionals in category-theoretical terms.
	\nr{%
		Conditionals ``only'' fall out as a special case at the end.
		The construction in this subsubsection is more general.
	}%
}%
\out{%
In combination with \smap, they expose task-level parallelism.
	\nr{It is clear that we are talking about task-level parallelism here.}
}%
\begin{align}
	\begin{array}{ll}
	f_1 : a \times s_{n_1} \rightarrow (b + c) \times s_{n_1} \,, &
		f_1^{*} : a \times s_N \rightarrow (b + c) \times s_N \,, \\[2pt]
	f_2 : b \times s_{n_2} \rightarrow b' \times s_{n_2} \,, &
		f_2^{*} : b \times s_N \rightarrow b' \times s_N \,, \\[2pt]
	f_3 : c \times s_{n_3} \rightarrow c' \times s_{n_3} \,, &
		f_3^{*} : c \times s_N \rightarrow c' \times s_N \,, \\[2pt]
	f_4 : (b'+ c') \times s_{n_4} \rightarrow d \times s_{n_4} \,, &
		f_4^{*} : (b'+ c') \times s_N \rightarrow d \times s_N \,.
	\end{array}
\end{align}
From the universal property of the coproduct, we obtain the state thread
\begin{align}
	&[f_2^{*}, f_3^{*}] : (b + c) \times s_N \rightarrow (b' + c') \times s_N \\
	&[f_2^{*}, f_3^{*}] (\inl \: x_b, (\sigma_{n_2}, \sigma_{n_3}, \sigma_{N\setminus\{n_1,n_2\}}))
		= \text{let } (x_b', \sigma_{n_2}') = f_2 (x_b, \sigma_{n_2}) \nonumber \\
	&\phantom{[f_2^{*}, f_3^{*}] (\inl \: x_b, (\sigma_{n_2}, \sigma_{n_3}, \sigma_{N\setminus\{n_1,n_2\}})) =}
		\text{ in } (\inl' \: x_b', (\sigma_{n_2}', \sigma_{n_3}, \sigma_{N\setminus\{n_1,n_2\}})) \\
	&[f_2^{*}, f_3^{*}] (\inr \: x_c, (\sigma_{n_2}, \sigma_{n_3}, \sigma_{N\setminus\{n_1,n_2\}}))
		= \text{let } (x_c', \sigma_{n_3}') = f_3 (x_c, \sigma_{n_3}) \nonumber \\
	&\phantom{[[f_2^{*}, f_3^{*}] (\inr \: x_c, (\sigma_{n_2}, \sigma_{n_3}, \sigma_{N\setminus\{n_1,n_2\}})) =}
		\text{in } (\inr' \: x_c', (\sigma_{n_2}, \sigma_{n_3}', \sigma_{N\setminus\{n_1,n_2\}})) \,.
\end{align}
\out{%
	\nr{%
		Omit this diagram for space reasons and for better page layout.
		Due to the absence of the state components $s_N$ in the diagram, the diagram is only a sketch anyway.
	}%
Figure~\ref{fig:universal:coproduct} gives a simplified version of the relevant universal diagram, where the state components $s_N$ have been omitted.
\begin{figure}
\begin{center}
	{\footnotesize
	\begin{tikzpicture}
		\matrix (m) [matrix of math nodes,row sep=3em,column sep=4em,minimum width=2em]
		{
			b & b  + c  &  c \\
			  & b' + c' &    \\};
		\path[-stealth]
		(m-1-1) edge node [above] {$\inl$} (m-1-2)
				edge node [left]  {$\inl' \circ f_2^{*}$} (m-2-2)
		(m-1-3) edge node [above] {$\inr$} (m-1-2)
				edge node [right] {$\inr' \circ f_3^{*}$} (m-2-2)
		(m-1-2) edge [dashed] node [right] {$[f_2^{*}, f_3^{*}]$} (m-2-2) ;
	\end{tikzpicture}
	}
\end{center}%
	\caption{Universal diagram for the coproduct $b+c$ with the natural injections \inl and \inr.}
	\label{fig:universal:coproduct}%
\end{figure}%
}
We can then form the composed state thread
\begin{align}
	f_4^{*} \circ [f_2^{*}, f_3^{*}] \circ f_1^{*} : a \times s_N \rightarrow d \times s_N \,.
\end{align}

To define how \smap acts on this state thread, we need two helper morphisms, \splt and \join, that use the boolean type $\mathbb{B}$ with values \textbf{T} and \textbf{F}.
The morphism \splt decomposes a list of coproduct values, i.e.~$[b + c]$, into two lists of types $[b]$ and $[c]$ respectively.
This decomposition is completely natural.
However, in defining the inverse operation of \splt, one faces a choice: %
The elements in the lists $[b]$ and $[c]$ can be arranged in different orders to form a list of coproduct values, $[b + c]$.
This choice introduces a source of non-determinism, which must be avoided since \ourLang is meant to be deterministic.
Therefore, \splt and \join operate on an additional data structure, namely a list of booleans, that encodes the order in which \join must form a list $[b + c]$ from the two lists $[b]$, $[c]$.
\begin{align}
& \splt : [b + c] \rightarrow [b] \times [c] \times [\mathbb{B}] \\
& \splt \left([]\right) = ([],[],[]) \\
& \splt \left((\inl \: x_b):\xs\right) = \plet (\bs, \cs, \flags) = \splt \: \xs \nonumber \\
& \phantom{\splt \left((\inl \: x_b):\xs\right) = \texttt{}} \pin (x_b:\bs, \,\cs, \,\textbf{T}:\flags) \\
& \splt \left((\inr \: x_c):\xs\right) = \plet (\bs, \cs, \flags) = \splt \: \xs \nonumber \\
& \phantom{\splt \left((\inr \: x_c):\xs\right) = \texttt{}} \pin (\bs, \,x_c:\cs, \,\textbf{F}:\flags) \\
\nonumber \\
& \join : [b] \times [c] \times [\mathbb{B}] \rightarrow [b + c] \\
& \join \left([],[],[]\right) = [] \\
& \join \left(x_b:\bs, \,\cs, \,\textbf{T}:\flags\right) = (\inl \: x_b) : \join\left(\bs, \cs, \flags\right) \\
& \join \left(\bs, \,x_c:\cs, \,\textbf{F}:\flags\right) = (\inr \: x_c) : \join\left(\bs, \cs, \flags\right)
\end{align}
The action of \smap on $f_4^{*} \circ [f_2^{*}, f_3^{*}] \circ f_1^{*}$ is then defined as follows,
\begin{align}
	&\smap\left( f_4^{*} \circ [f_2^{*}, f_3^{*}] \circ f_1^{*} \right) : [a] \times s_N \rightarrow [d] \times s_N \\
	&\smap\left( f_4^{*} \circ [f_2^{*}, f_3^{*}] \circ f_1^{*} \right) \:
		(\as, (\sigma_{n_1}, \sigma_{n_2}, \sigma_{n_3}, \sigma_{n_4}, \widetilde{\sigma})) = \nonumber \\
	&\quad\quad\text{let }
		(\as', (\sigma_{n_1}', \sigma_{n_2}, \sigma_{n_3}, \sigma_{n_4}, \widetilde{\sigma}))
			= \smap\left( f_1^{*} \right) \: (\as, (\sigma_{n_1}, \sigma_{n_2}, \sigma_{n_3}, \sigma_{n_4}, \widetilde{\sigma})) \nonumber \\
	&\quad\quad\phantom{\text{let }}
		(\bs, \cs, \textcolor{red}{\flags}) = \splt \: \as' \nonumber \\
	&\quad\quad\phantom{\text{let }}
		(\bs', (\sigma_{n_1}', \sigma_{n_2}', \sigma_{n_3}, \sigma_{n_4}, \widetilde{\sigma}))
			= \smap\left( f_2^{*} \right) \: (\bs, (\sigma_{n_1}', \sigma_{n_2}, \sigma_{n_3}, \sigma_{n_4}, \widetilde{\sigma})) \nonumber \\
	&\quad\quad\phantom{\text{let }}
		(\cs', (\sigma_{n_1}', \sigma_{n_2}, \sigma_{n_3}', \sigma_{n_4}, \widetilde{\sigma}))
		= \smap\left( f_3^{*} \right) \: (\cs, (\sigma_{n_1}', \sigma_{n_2}, \sigma_{n_3}, \sigma_{n_4}, \widetilde{\sigma})) \nonumber \\
	&\quad\quad\phantom{\text{let }}
		\ds = \join \: (\bs', \cs', \textcolor{red}{\flags}) \nonumber \\
	&\quad\quad{\text{in }}
		\smap\left( f_4^{*} \right) \: (\ds, (\sigma_{n_1}', \sigma_{n_2}', \sigma_{n_3}', \sigma_{n_4}, \widetilde{\sigma})) \,,
		\label{eq:smap:task:parallelism}
\end{align}
where $\widetilde{\sigma} \in s_{N\setminus\{n_1,n_2,n_3,n_4\}}$.
Note how $\textcolor{red}{\flags}$ is used to ensure determinism by communicating the order of list elements between \splt and \join.
Task-level parallelism can be utilized in Equation~\eqref{eq:smap:task:parallelism} by concurrently executing $\smap(f_2^{*})$ and $\smap(f_3^{*})$, which is possible since there are no dependencies between the data and state components operated on by $f_2^{*}$ and $f_3^{*}$.

A special case of the previous construction is obtained for $b' = d$, $c' = d$, and
\begin{align}
	&f_4 : (d + d) \times s_{n_4} \rightarrow d \times s_{n_4} \\
	&f_4 (\inl \: x_d, \sigma) = (x_d, \sigma) \\
	&f_4 (\inr \: x_d, \sigma) = (x_d, \sigma) \,.
\end{align}
With this $f_4$, $\textit{smap}(f_4^{*} \circ [f_2^{*}, f_3^{*}] \circ f_1^{*})$ yields a task-parallel version of an if-expression.
Conditionals with more than two options are obtained by repeatedly applying the construction from this section.
%
%
%
%
%
\subsection{Proof of Lemma~\ref{thm:lem:Phi:determines:Psi}}%
\label{sec:smap:well:defined}%
Recall that Lemma~\ref{thm:lem:Phi:determines:Psi} states that the functor $\Phi_{\calM}$ fully determines $\Psi_{\calM}$.
As a preliminary step towards establishing this, we derive a recursive formula for $\Psi_{\calM}$.
\begin{lemma}%
\label{thm:lem:Psi:formula}%
Let $w\in\morph{\freeCat{\Delta_{\calM}}}$ be such that no letter of $N$ occurs more than once in $w$.
Let $W$ be the set of letters in $w$, and let $\sigma_{W} \in s_{W}$, $\sigma_{N\setminus W} \in s_{N\setminus W}$.
Then,
\begin{align}
	& \Psi_{\calM}(w) \: (x : \xs, (\sigma_{W}, \sigma_{N\setminus W})) = \texttt{let }
		(y, (\sigma_{W}', \sigma_{N\setminus W}))
			= \Phi_{\calM}(w) \: (x, (\sigma_{W}, \sigma_{N\setminus W})) \nonumber \\
	& \phantom{\Psi_{\calM}(w) \: (x:\xs, (\sigma_{W}, \sigma_{N\setminus W})) = \texttt{let }}
		(\ys, (\sigma_{W}'', \sigma_{N\setminus W}))
			= \Psi_{\calM}(w) \: (\xs, (\sigma_{W}', \sigma_{N\setminus W})) \nonumber \\
	& \phantom{\Psi_{\calM}(w) \: (x:\xs, (\sigma_{W}, \sigma_{N\setminus W})) = \texttt{}} \texttt{in }
		(y:\ys, (\sigma_{W}'', \sigma_{N\setminus W})) \,.
	\label{eq:Psi:formula}
\end{align}
\begin{proof}
By induction on the length of $w$.
For $w = \epsilon_v$, $v\in\obj{\freeCat{\Delta_{\calM}}}$, Equation~\eqref{eq:Psi:formula} holds trivially since $\Psi_{\calM}(\epsilon_v) = \id$ and $\Phi_{\calM}(\epsilon_v) = \id$.
For the induction step, let $w = nw'$ with $n\in N$ and $w'\in\morph{\freeCat{\Delta_{\calM}}}$. 
Let $W'$ be the set of letters in $w'$, and let $\sigma_{W} = (\sigma_{n}, \sigma_{W'})$ with $\sigma_{n} \in s_{n}$, $\sigma_{W'} \in s_{W'}$.
Then,
\begin{flalign}
	&\Psi_{\calM}(nw') \: (x:\xs, (\sigma_{n},\sigma_{W'}, \sigma_{N\setminus W})) = 
		\psi_{\calM}^{*}(n) \circ \Psi_{\calM}(w') \: (x:\xs, (\sigma_{n},\sigma_{W'},\sigma_{N\setminus W})) & 
		\label{eq:Psi:formula:proof:1} \\
	& = \texttt{let }
		(y:\ys, (\sigma_{n},\sigma_{W'}'',\sigma_{N\setminus W})) 
			= \Psi_{\calM}(w') \: (x:\xs, (\sigma_{n},\sigma_{W'},\sigma_{N\setminus W})) \nonumber & \\
	& \phantom{ = }\texttt{ in }
		\psi_{\calM}^{*}(n) \: (y:\ys, (\sigma_{n},\sigma_{W'}'',\sigma_{N\setminus W})) & 
		\label{eq:Psi:formula:proof:2} \\
	& = \texttt{let }		
		(y, (\sigma_{n},\sigma_{W'}',\sigma_{N\setminus W}))
			= \Phi_{\calM}(w') \: (x, (\sigma_{n},\sigma_{W'},\sigma_{N\setminus W})) \nonumber \\
	& \phantom{ = \texttt{ let} }\:
		(\ys, (\sigma_{n},\sigma_{W'}'',\sigma_{N\setminus W}))
			= \Psi_{\calM}(w') \: (\xs, (\sigma_{n},\sigma_{W'}',\sigma_{N\setminus W})) \nonumber \\
	& \phantom{ = }\texttt{ in let }
		(z, (\sigma_{n}',\sigma_{W'}'',\sigma_{N\setminus W}))
			= \phi_{\calM}^{*}(n) (y, (\sigma_{n},\sigma_{W'}'',\sigma_{N\setminus W})) \nonumber \\
	& \phantom{ = \texttt{ in let}}\:
		(\zs, (\sigma_{n}'',\sigma_{W'}'',\sigma_{N\setminus W})) 
			= \psi_{\calM}^{*}(n) \: (\ys, (\sigma_{n}',\sigma_{W'}'',\sigma_{N\setminus W})) \nonumber \\
	& \phantom{ = \texttt{ in}}\:\texttt{in }
		(z:\zs, (\sigma_{n}'',\sigma_{W'}'',\sigma_{N\setminus W})) 
		\label{eq:Psi:formula:proof:3} \\
	& = \texttt{let }		
		(z, (\sigma_{n}',\sigma_{W'}',\sigma_{N\setminus W}))
			= \Phi_{\calM}(nw') \: (x, (\sigma_{n},\sigma_{W'},\sigma_{N\setminus W})) \nonumber \\
	& \phantom{ = \texttt{ let} }\:
		(\zs, (\sigma_{n}'',\sigma_{W'}'',\sigma_{N\setminus W}))
			= \Psi_{\calM}(nw') \: (\xs, (\sigma_{n}',\sigma_{W'}',\sigma_{N\setminus W})) \nonumber \\
	& \phantom{ = }\texttt{ in }
		(z:\zs, (\sigma_{n}'',\sigma_{W'}'',\sigma_{N\setminus W})) \,,
		\label{eq:Psi:formula:proof:4}
\end{flalign}%
where the induction hypothesis was used in going from Equation~\eqref{eq:Psi:formula:proof:2} to Equation~\eqref{eq:Psi:formula:proof:3}.
The manipulation required to go from Equation~\eqref{eq:Psi:formula:proof:3} to Equation~\eqref{eq:Psi:formula:proof:4} is  known as \emph{let floating} in the context of functional language compilers~\cite{PeytonJones:1996:LMB:232627.232630}.
The assumption that no letter occurs more than once in $w=nw'$ is used whenever elements of state objects are decomposed into components and to determine on which of these components $\Phi_{\calM}$ and $\Psi_{\calM}$ act as the identity.
\end{proof}%
\end{lemma}%
\begin{proof}[Proof of Lemma~\ref{thm:lem:Phi:determines:Psi}]
Let $a = \src{w_1} = \src{w_2}$.
Let $\xs\in [a]$ and let $\sigma\in s_N$.
The proof proceeds by induction on the length of \xs.
For $\xs = []$, one finds immediately that $\Psi_{\calM}(w_1) \: ([],\sigma) = ([],\sigma) = \Psi_{\calM}(w_2) \: ([],\sigma)$.
Now, let $\xs = x:\xs'$, with $x\in a$, $\xs'\in [a]$.
From Lemma~\ref{thm:lem:Psi:formula},
\begin{align}
	& \Psi_{\calM}(w_1) \: (x:\xs', \sigma) = \texttt{let }
		(y, \sigma')
			= \Phi_{\calM}(w_1) \: (x, \sigma) \nonumber \\
	& \phantom{\Psi_{\calM}(w_1) \: (x : \xs', \sigma) = \texttt{let }}
		(\ys, \sigma'')
			= \Psi_{\calM}(w_1) \: (\xs, \sigma') \nonumber \\
	& \phantom{\Psi_{\calM}(w_1) \: (x:\xs', \sigma) = \texttt{}} \texttt{in }
		(y:\ys, \sigma'')
		\label{eq:Phi:determines:Psi:proof:1} \\
	& \phantom{\Psi_{\calM}(w_1) \: (x:\xs', \sigma) } = \texttt{let }
		(y, \sigma')
			= \Phi_{\calM}(w_2) \: (x, \sigma) \nonumber \\
	& \phantom{\Psi_{\calM}(w_1) \: (x:\xs', \sigma) = \texttt{let }}
		(\ys, \sigma'')
			= \Psi_{\calM}(w_2) \: (\xs, \sigma') \nonumber \\
	& \phantom{\Psi_{\calM}(w_1) \: (x:\xs', \sigma) = \texttt{}} \texttt{in }
		(y:\ys, \sigma'') 
		\label{eq:Phi:determines:Psi:proof:2} \\
	& \phantom{\Psi_{\calM}(w_1) \: (x:\xs', \sigma)} = 
		\Psi_{\calM}(w_2) \: (x:\xs', \sigma) \,.
		\label{eq:Phi:determines:Psi:proof:3}
\end{align}
Going from Equation~\eqref{eq:Phi:determines:Psi:proof:1} to Equation~\eqref{eq:Phi:determines:Psi:proof:2} uses both the assumption $\Phi_{\calM}(w_1) = \Phi_{\calM}(w_2)$ and the induction hypothesis.
Equation~\eqref{eq:Phi:determines:Psi:proof:3} is arrived at by applying Lemma~\ref{thm:lem:Psi:formula} again.
\end{proof}%



\out{%
	\nr{%
		The current presentation allows arbitrary graphs up to this point.
		Therefore, the discussion in the following two paragraphs is not relevant here, and it can probably be removed from the paper altogether.
	}%
This has a consequence that have not been looked at so far:
the set \calM cannot contain any morphisms $f : (a\times s_n) \rightarrow (a\times s_n)$ since this would require the presence of a loop in $\Delta^N_{\calO}$.
To still include morphisms like $f : (a\times s_n) \rightarrow (a\times s_n)$ in our programming model, two options exist:
(i) one could consider the elements of $\calO$ as labels on the independently existing vertices of $\Delta^N_{\calO}$.
Then, different vertices can have the same label, and hence edges $a\rightarrow a$ are possible without introducing loops into the DAG.
(ii) Alternatively, one can allow elements of \calO to be type aliases of each other.
Thus, if $a, b$ are aliases for the same type, then an edge $a\rightarrow b$ in $\Delta^N_{\calO}$ corresponds to a morphism $a\rightarrow a$ under the aliasing relationship between the types $a$ and $b$.

The discussion so far has not made use of the acyclicity of $\Delta^N_{\calO}$.
In fact, the theory can be developed equally well when cycles exist.
However, as we will see in the next section, it is precisely the absence of cycles that allows for an implementation that takes advantage of pipeline parallelism.
}


\begin{acks}                            
This work was supported in part by the German Research Foundation (DFG) within the Collaborative Research Center HAEC and the Center for Advancing Electronics Dresden (cfaed).
\end{acks}

\bibliography{io}


\begin{thebibliography}{00}


\ifx \showCODEN    \undefined \def \showCODEN     #1{\unskip}     \fi
\ifx \showDOI      \undefined \def \showDOI       #1{{\tt DOI:}\penalty0{#1}\ }
  \fi
\ifx \showISBNx    \undefined \def \showISBNx     #1{\unskip}     \fi
\ifx \showISBNxiii \undefined \def \showISBNxiii  #1{\unskip}     \fi
\ifx \showISSN     \undefined \def \showISSN      #1{\unskip}     \fi
\ifx \showLCCN     \undefined \def \showLCCN      #1{\unskip}     \fi
\ifx \shownote     \undefined \def \shownote      #1{#1}          \fi
\ifx \showarticletitle \undefined \def \showarticletitle #1{#1}   \fi
\ifx \showURL      \undefined \def \showURL       #1{#1}          \fi
\providecommand\bibfield[2]{#2}
\providecommand\bibinfo[2]{#2}
\providecommand\natexlab[1]{#1}
\providecommand\showeprint[2][]{arXiv:#2}

\bibitem[\protect\citeauthoryear{Ertel, Adam, Rink, Goens, and
  Castrillon}{Ertel et~al\mbox{.}}{2019}]%
        {ertel_haskell19}
\bibfield{author}{\bibinfo{person}{Sebastian Ertel}, \bibinfo{person}{Justus
  Adam}, \bibinfo{person}{Norman~A. Rink}, \bibinfo{person}{Andr{\'e}s Goens},
  {and} \bibinfo{person}{Jeronimo Castrillon}.}
  \bibinfo{year}{2019}\natexlab{}.
\newblock \showarticletitle{STCLang: State Thread Composition As a Foundation
  for Monadic Dataflow Parallelism}. In \bibinfo{booktitle}{{\em Proceedings of
  the 12th ACM SIGPLAN International Symposium on Haskell}} {\em
  (\bibinfo{series}{Haskell 2019})}. \bibinfo{publisher}{ACM},
  \bibinfo{address}{New York, NY, USA}, \bibinfo{pages}{146--161}.
\newblock
\showISBNx{978-1-4503-6813-1}
\showDOI{%
\url{http://dx.doi.org/10.1145/3331545.3342600}}


\bibitem[\protect\citeauthoryear{Peyton~Jones, Partain, and
  Santos}{Peyton~Jones et~al\mbox{.}}{1996}]%
        {PeytonJones:1996:LMB:232627.232630}
\bibfield{author}{\bibinfo{person}{Simon Peyton~Jones}, \bibinfo{person}{Will
  Partain}, {and} \bibinfo{person}{Andr{\'e} Santos}.}
  \bibinfo{year}{1996}\natexlab{}.
\newblock \showarticletitle{Let-floating: Moving Bindings to Give Faster
  Programs}. In \bibinfo{booktitle}{{\em Proceedings of the First ACM SIGPLAN
  International Conference on Functional Programming}} {\em
  (\bibinfo{series}{ICFP '96})}. \bibinfo{publisher}{ACM},
  \bibinfo{address}{New York, NY, USA}, \bibinfo{pages}{1--12}.
\newblock
\showISBNx{0-89791-770-7}
\showDOI{%
\url{http://dx.doi.org/10.1145/232627.232630}}


\end{thebibliography}

\end{document}